\documentclass[aps,prx,twocolumn,superscriptaddress]{revtex4}
\usepackage{graphicx,amsmath,amssymb,amsthm,enumerate,yfonts}
\newtheorem{thm}{Theorem}
\newtheorem{lem}{Lemma}
\newtheorem{cor}{Corollary}

\usepackage{tikz}
\usetikzlibrary{arrows,calc}
\begin{document}
\title{Symmetric $M$-ary phase discrimination using quantum-optical probe states}
\author{Ranjith Nair}
\affiliation{Department of Electrical and Computer Engineering, National University of Singapore, 4 Engineering Drive 3, Singapore 117583}
\author{Brent J.~Yen}
\affiliation{Department of Electrical and Computer Engineering, National University of Singapore, 4 Engineering Drive 3, Singapore 117583}
\author{Saikat Guha}
\affiliation{Disruptive Information Processing Technologies Group, Raytheon BBN Technologies, Cambridge, Massachusetts 02138, USA}
\author{Jeffrey H.~Shapiro} \affiliation{Research Laboratory of Electronics, Massachusetts Institute of Technology, Cambridge, Massachusetts 02139, USA}
\author{Stefano Pirandola}
\affiliation{Department of Computer Science, University of York, York YO10 5GH, United Kingdom}
\date{June 4, 2012}
\newcommand\mbf[1]{\mathbf{#1}}
\newcommand\ovl[1]{\overline{#1}}
\newcommand\udl[1]{\underline{#1}}
\newcommand {\ket}[1] {|{#1}\rangle}
\newcommand {\kets} [1] {|#1\rangle_{S}}
\newcommand {\ketr} [1] {|#1\rangle_{R}}
\newcommand {\keti} [1] {|#1\rangle_{I}}
\newcommand {\ketsi} [1] {|#1\rangle_{IS}}
\newcommand {\ketri} [1] {|#1\rangle_{IR}}
\newcommand {\brasi} [1] {\hspace{0mm}_{IS}\langle#1|}
\newcommand {\brari} [1] {\hspace{0mm}_{IR}\langle#1|}
\newcommand {\bras} [1] {\hspace{0mm}_{S}\langle#1|}
\newcommand {\brai} [1] {\hspace{0mm}_{I}\langle#1|}
\newcommand {\bra}[1] {\langle{#1}|}
\newcommand{\braket}[2]{\langle{#1}|{#2}\rangle}
\newcommand{\brakets}[2]{\hspace{0mm}_{S}\langle#1|#2\rangle_{S}}
\newcommand{\braketr}[2]{\hspace{0mm}_{R}\left\langle#1\,|\,#2\right\rangle_{R}}
\newcommand{\braketis}[2]{\hspace{0mm}_{IS}\langle#1|#2\rangle_{IS}}
\newcommand{\braketir}[2]{\hspace{0mm}_{IR}\langle#1|#2\rangle_{IR}}
\newcommand {\cl}{\mathcal}
\newcommand {\tsf} [1]{\textsf{#1}}
\newcommand {\norm} [1] {\parallel #1 \parallel}
\newcommand{\tr} {\tsf{tr}}

\begin{abstract}
We present a theoretical study of minimum error probability discrimination, using quantum-optical probe states, of $M$ optical phase shifts situated symmetrically on the unit circle. We assume ideal lossless conditions and full freedom for implementing quantum measurements and for probe state selection, subject only to a constraint on the average energy, i.e., photon number. In particular, the probe state is allowed to have any number of signal and ancillary modes, and to be pure or mixed. Our results are based on a simple criterion that partitions the set of pure probe states into equivalence classes with the same error probability performance. Under an energy constraint, we find the explicit form of the state that minimizes the error probability. This state is an unentangled but nonclassical single-mode state. The error performance of the optimal state is compared with several standard states in quantum optics. We also show that discrimination with zero error is possible only beyond a threshold energy of $(M-1)/2$. For the $M=2$ case, we show that the optimum performance is readily demonstrable with current technology. While transmission loss and detector inefficiencies lead to a nonzero erasure probability, the error rate conditional on no erasure is shown to remain the same as the optimal lossless error rate. 
\end{abstract}
\maketitle

\section{Introduction}
The estimation of an optical phase shift using quantum states of light is a well-known theme of both theoretical and experimental studies, and is still an active area of research, as seen from a sample \cite{phaseestimationreferences} of recent work. A less well-known but analogous sensing problem is that of discriminating a finite number $M \geq 2$ of phase shifts  symmetrically arranged on the unit circle, which may be thought of as a discrete version of the phase estimation problem. The problems differ in the criterion used to measure performance -- in the estimation problem, a typical figure of merit is the mean squared error while an error probability criterion is natural for the discrimination problem. 

The phase discrimination problem may be viewed in communication terms in analogy with $M$-ary phase shift keying (PSK) in ordinary and optical communications \cite{Proakis07,Agrawal10}.  In $M$-ary PSK, one of $M$ uniformly spaced phase shifts is applied to a predetermined waveform for the purpose of communicating one of $M$ messages (or, equivalently, $\log_2 M$ bits of data) from the sender to the receiver. Binary PSK (BPSK) using coherent states of light, with error probabilities near the quantum limit in the absence of noise, is already a rather mature technology \cite{GnauckWinzer05}. As the demand for high-speed communication increases, $M$-ary PSK with $M>2$ is becoming attractive in optical communication, despite the increased system complexity, because it provides an increase in the number of bits per transmitted symbol without an increase in the frequency bandwidth required \cite{Agrawal10}. In the quantum version of PSK (see, e.g., the studies \cite{HallFuss91,Shapiro93,Kato99,Becerraetal11}) that is the subject of this paper, an optical mode prepared in a predetermined quantum state is the analog of the classical signal waveform to which the information-bearing phase shifts are applied. 

Another application of our study is to the recently developed concept of quantum reading of a classical digital memory \cite{Pirandola11,Nair11,Hirota11,Pirandolaetal11,Guhaetal11,Weedbrooketal11}. The original proposal of Ref.~\cite{Pirandola11} considered the use of a quantum-optical probe state that reads a standard optically encoded digital memory such as a CD or DVD with a bit error probability better than that achievable with standard laser sources. In a CD or DVD, information is stored by varying the properties of the recording surface in a rather involved manner -- see, e.g., Ref. \cite{Brooker03}. However, the overall reading process may be modeled as the discrimination of two beam-splitter channels with transmittance depending on the data bit, i.e., the model involves an amplitude (rather than phase) encoding. In Ref.~\cite{Nair11}, a general problem of discriminating two beam-splitter channels was analyzed that models any kind of bit encoding, either in phase or amplitude or a combination of the two. 

The quantum reading and beam-splitter discrimination problems of Refs.~\cite{Pirandola11} and \cite{Nair11} include optical loss naturally and correspond mathematically to the discrimination of non-unitary quantum channels. In Ref.~\cite{Hirota11}, a purely phase-encoded memory, which can be lossless in principle, was proposed. This memory encodes a  bit $0\, (1)$ as a $0\, (\pi)$  radian phase shift imparted directly to a probe beam upon reflection from the encoding surface. As such, it corresponds exactly to the $M=2$ case of the problem considered in this paper.  Some variants of quantum reading that correspond to discrimination of unitary channels have been proposed and also experimentally demonstrated \cite{Bisio11,Dall'Arno11}. Quantum reading of an amplitude-encoded memory has been studied from an information-theoretic perspective in refs.~\cite{Pirandolaetal11,Guhaetal11}.
 
Viewed as a problem of quantum decision theory, the phase discrimination problem falls under the general rubric of distinguishing a \emph{symmetric set} of quantum states that has been  studied extensively \cite{Belavkin75,Helstrom76,Ban97,Holevo79,HausladenWootters94,Hausladenetal96}. The optimal quantum measurement for a given probe state was obtained in the pioneering works \cite{Belavkin75,Helstrom76,Ban97}. This optimal measurement, called variously as the \emph{Square-Root Measurement} (SRM), \emph{Least-Squares Measurement}, and \emph{Pretty Good Measurement}, has many interesting properties and has also been applied to quantum information theory \cite{Holevo79,HausladenWootters94,Hausladenetal96}. 

In this paper, our concern is mainly with the \emph{design problem} of choosing the best (i.e., yielding the least error probability) probe state under a given energy constraint.  This problem has not been addressed in full generality in the literature, although related studies of quantum communication using $M$-ary PSK exist. Thus, for single-mode probe states, the problem of probe state optimization under an energy constraint was introduced and studied  in ref.~\cite{HallFuss91}, but under a restricted class of allowed measurements. The possibility of zero-error communication was mentioned in Ref.~\cite{HallFuss91} for single-mode states, while in ref.~\cite{Shapiro93}, it was established that two-mode phase-conjugate PSK achieves the same end. The SRM was used to give a detailed performance evaluation of $M$-ary PSK for the case of the practically important coherent-state transmitters in ref.~\cite{Kato99}.  Very recently, a receiver that achieves near-optimal error probability for coherent-state $M$-PSK was proposed \cite{Becerraetal11} and demonstrated for 4-PSK. However, a fully general treatment of the $M$-ary phase discrimination problem is as yet lacking, even in the ideal lossless case. In this paper, we present such an analysis and obtain the best probe state over all quantum states of any number of signal and ancillary modes under an average energy constraint. 

This paper is organized as follows. In Section II, we set up the mathematical model of the symmetric phase discrimination problem along with a discussion of the physical assumptions involved. Section III elaborates its solution as follows. In Sections III.A-III.C, we consider the case of a pure-state probe. In Section III.A,  we obtain the optimal state under an energy constraint on the modes experiencing the phase shift. In Section III.B, we briefly discuss the associated quantum measurement that achieves the minimum error probability. In Section III.C, we consider the case of a combined energy constraint on all the probe-state modes. In Section III.D,  the optimality proof under both energy constraints is extended to allow for mixed-state probes. In Section III.E, we  present performance curves for several standard states in quantum optics alongside those of the optimal state and the coherent states. In Section IV, we present an easy implementation of $M=2$ phase discrimination under both lossless and lossy conditions. In Section V, we conclude by discussing some possible future directions based on the present work.

\section{Symmetric Phase Discrimination: Problem Setup}

\subsection{Probe state}

Consider, for an integer $M \geq 2$, for $\theta_M := 2\pi/M$, and for  $\mathbb{Z}_M :=  \{0, \ldots, M-1\}$, the set $\{\theta_m := m\, \theta_M \, |\, m \in \mathbb{Z}_M\}$ of phase shifts symmetrically disposed on the unit circle as in Fig.~1.
\begin{figure}
\includegraphics[trim= 50mm 165mm 45mm 50mm, clip=true, width=0.5\textwidth]{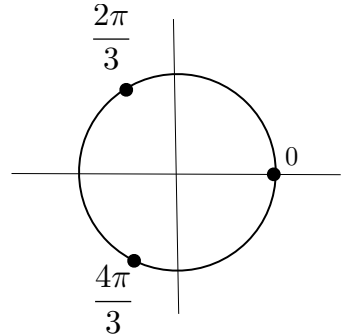}
\caption{Symmetric phase shifts on the unit circle for $M=3$.}
\end{figure}
 The application of, say, the $m$-th such phase shift to each of $J\geq1$ quasi-monochromatic optical field modes (with annihilation operators $\{\hat{a}_S^{(j)}\}_{j=1}^{J}$) is represented by the unitary operator
\begin{align} \label{unitary.m}
  \hat{U}_m = \bigotimes_{j=1}^{J} e^{i m \theta_M \hat{N}_S^{(j)}} \equiv \bigotimes_{j=1}^{J} \hat{U}_m^{(j)},\
      m \in  \mathbb{Z}_M,
\end{align}
where $\hat{N}_S^{(j)} = \hat{a}_S^{(j)\dag}\hat{a}_S^{(j)}$ is the number operator of the $j$-th mode, with $1 \leq j \leq J$. These modes that undergo the phase shift are called \emph{signal} modes, indicated by the subscript `S' (see Fig.~2).  In addition to the $J$ signal modes, we also allow, as depicted in Fig.~2, any number $J'\geq 0$ of ancillary modes. These  are called \emph{idler} modes, indicated by the subscript `I', and have annihilation operators $\{\hat{a}_I^{(j')}\}_{j'=1}^{J'}$. The idler modes do not acquire the $m$-dependent phase shift, but allow for the preparation of a quantum state that is entangled across the signal and idler modes.  Such a joint state on the signal and idler modes will be called a \emph{probe state}. We are interested in the problem of the choice of probe state that minimizes the error probability in determining $m$ when the latter is drawn, unknown to the receiver, at random from $\mathbb{Z}_M$. 
\begin{figure}
\begin{tikzpicture}[font=\small,>=stealth']
  \node at (0,0.8)         (S) {$\hat{a}_S^{(j)}$};
  \node at (0,0)           (I) {$\hat{a}_I^{(j')}$};
  \node[draw] at (1.8,1.8) (U) {$\hat{U}_m^{(j)}$};
  \node at (4,0.8)         (S2) {$\hat{a}_R^{(j)}$};
  \node at (4,0)           (I2) {$\hat{a}_I^{(j')}$};
  \draw[->] (S) -- (U);
  \draw[->] (U) -- (S2);
  \draw[->] (I) -- (I2);
\end{tikzpicture}
\caption{A  pure state $|\psi\rangle_{IS}$ of $J$ signal modes (represented by the annihilation operators $\{\hat{a}_S^{(j)}\}_{j=1}^J$) and $J'$ idler modes (represented by the annihilation operators $\{\hat{a}_I^{(j')}\}_{j'=1}^{J'}$) is prepared. The signal modes pass through a phase-shifting element that modulates the phase of the incident light via one of the unitary transformations $\hat{U}_m$ specified by Eq.~\eqref{unitary.m}.  The return modes (represented by the annihilation operators $\{\hat{a}_R^{(j)}\}_{j=1}^J$) and
idler modes, the latter remaining unaffected by the phase shift, are measured using a minimum error probability quantum measurement.}
\end{figure}
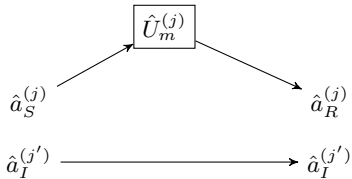

An arbitrary pure probe state of $J+J'$ modes may be written in the multimode number basis as
\begin{align}
\label{probe}
  \ket{\psi}_{IS}= \sum_{\mbf{k},\mbf{n}} c_{\mbf{k},\mbf{n}} \keti{\mbf{k}} \kets{\mbf{n}},
\end{align}
where $|\mbf{k}\rangle_I = |k_1\rangle\otimes\dotsb\otimes|k_{J'}\rangle$
and $|\mbf{n}\rangle_S = |n_1\rangle\otimes\dotsb\otimes|n_J\rangle$ are multimode Fock states.
For any such probe state, the output states for the idler and \emph{return} (`R') modes are
\begin{align} \label{output states} 
\ketri{\psi_m} = \hat{I}_I\otimes \hat{U}_m  \ketsi{\psi}, \hspace{4mm} m \in \mathbb{Z}_M,
\end{align}
where $\hat{I}_I$ is the identity transformation on the idler modes. The set of complex numbers $\{e^{im\theta_M} |\hspace{1mm}m \in \mathbb{Z}_M\}$ is a cyclic group of order $M$ under multiplication and consequently, so is the set of $M$ unitary operators $\{\hat{I}_I \otimes \hat{U}_m\, |\, m \in \mathbb{Z}_M\}$. Assuming equal \emph{a priori} probabilities for the $M$ unitaries, the output states of Eq.~\eqref{output states} satisfy the \emph{symmetric set} condition  of ref.~\cite{Ban97} defined by
\begin{align} \label{symmetric set condition.I}
\ket{\psi_m} &= \hat{V}^m \ket{\psi}, \hspace{4mm}m \in \mathbb{Z}_M, \\
\hat{V}^M &= \hat{I} \label{symmetric set condition.II},
\end{align}
for some unitary operator $\hat{V}$ and some seed state $\ket{\psi}$ \cite{footnote1}. The correspondence to our problem is obtained by taking the seed state $\ket{\psi}$ in \eqref{symmetric set condition.I} to be the probe state $\ketsi{\psi}$ and the generating unitary operator $\hat{V}$ of eqs.~(\ref{symmetric set condition.I}) and (\ref{symmetric set condition.II}) to be
\begin{align}
\hat{V}  = \hat{I}_I \otimes \bigotimes_{j=1}^{J}e^{i  \theta_M \hat{N}_S^{(j)}} = \hat{I}_I \otimes \hat{U}_1
\end{align}
with $\hat{U}_1$ given by eq.~\eqref{unitary.m}.

The discrimination strategy of Fig.~2 may be called an \emph{entanglement-assisted parallel strategy} in analogy with the terminology of  ref.~\cite{GLM11}, and corresponds to Fig.~2d of \cite{GLM11} augmented with idler modes. This is clearly not the most general strategy. For example, we may consider \emph{sequential strategies} (see Fig.~$3$ of \cite{GLM11}). It is easy to show that the simple sequential strategy depicted in Fig.~$3a$ of \cite{GLM11} cannot help, as follows. Successive application of two phase shifts of $m \theta_M$ rad results in a phase shift $2 m \theta_M$ rad which is also in the set  $\{m \theta_M\, |\, m \in \mathbb{Z}_M\}$ due to the group property of $\mathbb{Z}_M$. Thus, at best, the set of output states after two applications of the phase shift is a permutation of the set of states after one application of the phase shift, and in general, the set of output states is even less distinguishable since distinct phase shifts may result in the same state after multiple applications even if they do not after a single application. That said, we have not ruled out the efficacy of more complicated strategies, such as the sequential strategy depicted in Fig.~$3b$ of \cite{GLM11} or strategies that adaptively select the input states of later signal modes conditioned on measurement results from earlier modes (see, e.g., \cite{Harrow10}). Optimization over all such strategies appears to be an involved task and will not be considered in this paper.

It is worthwhile to mention a few more features of our problem setup. First, our model of the phase discrimination problem, as given by Eq.~\eqref{output states}, assumes the presence of a phase reference in that we have the ability to prepare pure probe states and not just phase-averaged mixed states. Such a phase reference may be physically implemented, for example, by a separate strong coherent-state beam. Second, though given a phase reference, we may still prepare arbitrary mixed-state probes that are not represented in Eq.~\eqref{output states} but must be considered in a fully general optimization. Nevertheless, it will be convenient to first address the pure-state probe case of Eq.~\eqref{output states} in detail. The mixed state case will be subsequently addressed in Section III.D to extend the optimization to the full state space.

\subsection{Probe state design}

Accordingly, we are now interested in minimizing the error probability over a set of allowed pure probe states $|\psi\rangle_{IS}$. This logically entails two successive minimizations. For a chosen probe state $\ket{\psi}_{IS}$, the minimum average error probability achievable in the sense of Helstrom \cite{Helstrom76} is given by
\begin{align} \label{error probability}
\overline{P}_e = 1 - \frac{1}{M} \max_{\{\hat{E}_m\}}\sum_{m=0}^{M-1} \mathrm{tr} \left(\ket{\psi_m} \brari{\psi_m} \,\hat{E}_m\right),
\end{align}
where $\{\hat{E}_m\}_{m=0}^{M-1}$ is a set of positive semidefinite operators constituting a positive-operator-valued measure (POVM) \cite{Helstrom76} that represents any quantum measurement process that discriminates the $M$ possibilities. The optimization over all POVM's in eq.~\eqref{error probability} corresponds to the choice of the optimum measurement process for a given probe state $\ket{\psi}_{IS}$. We are interested in minimizing \eqref{error probability} further over a set of `allowed' probe states as detailed below. 

Intuitively, it is clear that, if no further restrictions are imposed, one may be able to achieve arbitrarily small error probability by choosing a probe state with sufficiently high energy. For example, we may use a single-mode coherent state probe $\ket{\sqrt{N_S}}$ of average photon number $N_S$, leading to the output states $\{\ket{\sqrt{N_S} e^{im\theta_M}}\}_{m=0}^{M-1}$. As $N_S$ is increased, the output states become more and more orthogonal and the error probability decreases. In the limit $N_S \rightarrow \infty$, we have $\overline{P}_e \rightarrow 0$. Thus, further constraints are required to make the optimization problem meaningful. In the communications literature, an energy constraint is typically imposed in addition to a limit on the number of modes $J$, i.e., the bandwidth. In the present context, an energy constraint is particularly important because it is practically hard to prepare novel quantum states with a large average photon number. One can also imagine scenarios, e.g., that of probing a sensitive biological sample, where a signal energy constraint must be imposed to avoid damaging the sample during probing. In view of the above considerations, we constrain the average total photon number in the signal modes
\begin{align} 
\label{signalenergyconstraint}
\langle\hat{N}_S\rangle  \equiv \left\langle \sum_{j=1}^J \hat{N}_S^{(j)} \right\rangle \leq N_S,
\end{align}
$N_S$ being a given number. The case of a constraint on the average photon number of the signal \emph{plus} idler modes is considered in Section III.C. For brevity, we will simply write `signal energy' for the average total photon number in the signal modes given by the left-hand side of Eq.~(\ref{signalenergyconstraint}).

Using eq.~\eqref{probe}, we calculate the signal energy
\begin{align}
 \langle\hat{N}_S\rangle
   &= \sum_{\mbf{k},\mbf{n}} \left(n_1+\cdots+ n_J\right) \left|c_{\mbf{k},\mbf{n}}\right|^2\\
   &=\sum_{\mbf{n}} \left( (n_1+\cdots+ n_J) \sum_{\mbf{k}} \left|c_{\mbf{k},\mbf{n}}\right|^2 \right) \\
   &\equiv \sum_{\mbf{n}} \,(n_1+\cdots+ n_J)\, p_{\mbf{n}}\\
   &\equiv \sum_{n=0}^{\infty} n \,p_n, \label{signalenergy}
\end{align}
where $p_n$ is the probability that the \emph{total} photon number $ n_1 +\cdots+ n_J$ in
the signal modes is $n$. We denote the signal photon probability distribution by the infinite-dimensional vector $\mathbf{p}= (p_0,p_1, \ldots)$. 

In the sequel, another discrete probability distribution derived from the probe state plays a major role. It depends on both the probe state (through $\mathbf{p}$) and on $M$, as follows. For $\nu \in \mathbb{Z}_M$, let $\textswab{p}\equiv(\textfrak{p}_0, \ldots, \textfrak{p}_\nu,\ldots, \textfrak{p}_{M-1})$ be defined component-wise as
\begin{align} \label{swabpdef}
\textfrak{p}_\nu := \sum_{\begin{array}{c} n :  n \equiv \nu\hspace{1mm}
 (\bmod \hspace{1mm}M) \end{array}} p_n.
\end{align}
In other words, $\textswab{p}$ is the probability distribution induced by $\mathbf{p}$ on the modulo-$M$ congruence classes of the total signal photon number.

\section{Symmetric phase discrimination: problem solution}

\subsection{Pure probe states}

It will be revealing to approach the problem set up in Section II in stages. We will comment as we go along on the implications of our results and their connections to the literature. As in Section II, we continue to assume a pure probe state and address the use of mixed probe states in Section III.D.   

We first present the basic result that, without invoking a signal energy constraint or specifying $J$ and $J'$, we can partition the probe state space into equivalence classes of states having the same error probability.
\begin{thm} \label{Thm:1}
Pure probe states with the same $\textswab{p}$ have the same performance in the discrimination of $M$ symmetric phases. This statement encompasses probes with differing $J$ and/or $J'$.
\begin{proof}
For an arbitrary probe state $\ketsi{\psi}$ written in the form of Eq.~\eqref{probe}, the corresponding output states $\{\ketri{\psi_m}\}_{m=0}^{M-1}$ are given by
\begin{align}
\ketri{\psi_m} =  \sum_{\mbf{k},\mbf{n}} c_{\mbf{k},\mbf{n}} e^{im\theta_M(n_1 + \cdots + n_J)}\keti{\mbf{k}} \ketr{\mbf{n}}.
\end{align}
Consider the $M \times M$ matrix $\mbf{G}$ (the \emph{Gram matrix}) whose elements are all the mutual inner products between the $\{\ketri{\psi_m}\}$, i.e., 
\begin{align}
G_{m\,m'} := \braketir{\psi_m}{\psi_{m'}}.
\end{align}
The minimum error probability \eqref{error probability} in discriminating the symmetric set of pure states $\{\ketri{\psi_m}\}_{m=0}^{M-1}$ is a function of the elements of $\mathbf{G}$ alone \cite{footnote2}. 
We compute the general element of $\mathbf{G}$ to be
\begin{align} G_{m\,m'} =& \braketir{\psi_m}{\psi_{m'}}\\
=& \sum_{\mbf{k},\mbf{n}}  \left|c_{\mbf{k},\mbf{n}}\right|^2\, e^{-i\theta_M(m-m') (n_1+\cdots+n_J)} \\
=& \sum_{\mbf{n}}  p_{\mbf{n}} \,e^{-i\theta_M(m-m') (n_1+\cdots+n_J)}\\
=& \sum_{n=0}^\infty p_n \, e^{-i\theta_M(m-m') n}  \\
=& \sum_{\nu=0}^{M-1} \textfrak{p}_\nu\, e^{-i\theta_M(m-m') \nu} . \label{grammatrixelement}
\end{align}
The equality \eqref{grammatrixelement} follows because, for any $m$ and $m'$, the exponential factor is
periodic in $n$ with period $M$.  We have thus shown that the Gram matrix, and thence the error performance, is a function of just the $M$ components of $\textswab{p}$, and that this is true irrespective of the values of $J$ and $J'$. 
\end{proof}
\end{thm}
The result of Theorem $1$ is interesting for a number of reasons. First, it clusters probe states into classes with the same error performance based on the easily computed characteristic $\textswab{p}$. Second, since one can always prepare a signal-only ($J'=0$) probe with a given $\textswab{p}$, the ancillary idler modes shown in Fig.~2 do not improve performance. This is unlike the typical situation in which ancillary entanglement in the probe helps in distinguishing $M$ unitary transformations \cite{D'Ariano01}. Third, since any given $\textswab{p}$ can be realized using a single-mode signal state (i.e., with $J=1$ as in Theorem $2$ below), no performance gain accrues from using multiple signal modes. This can again be contrasted with the situation of distinguishing two general finite-dimensional unitaries, for which multiple applications of the unitaries can result in error-free discrimination \cite{Acin01}. We mention that these latter two implications of Theorem 1 also follow from a general lossless image sensing result of ref.~\cite{NairYen11} (see section on lossless imaging therein). Finally, the freedom of probe state choice allowed by Theorem 1 will turn out (in Section IV) to be crucial to a practical implementation of the $M=2$ case.

Our next result implies that, when the signal energy is constrained as in \eqref{signalenergyconstraint} to a maximum of $N_S$, the most efficient way to use the energy is to use a probe state with $\mathbf{p}$ supported on just its first $M$ components. It is also shown that, beyond a threshold signal energy of $(M-1)/2$, discrimination with \emph{zero error} is possible.
\begin{thm} 
\label{Thm:2} 
 \begin{enumerate}
 [(a)] \item For $N_S <(M-1)/2$, a single-mode probe state of the form $\kets{\psi}=\sum_{\nu=0}^{M-1} \sqrt{\textfrak{p}_{\nu}}\,\kets{\nu}$ with $\textfrak{p}_{\nu} \geq 0$ achieves the minimum error probability. 

\item For $N_S\geq (M-1)/2$, the uniform superposition state
$\kets{\psi}= \frac{1}{\sqrt{M}}\left(\kets{0} + \cdots + \kets{M-1}\right)$ achieves perfect discrimination. 
\end{enumerate}
\begin{proof}
(\textit{a}) Presented with any probe state $\ketsi{\psi}$ with associated $\mathbf{p}$ and $\textswab{p}$, we can construct the
single-mode probe
\begin{align}
\label{singlemodepsi}
  \kets{\psi} = \sum_{\nu=0}^{M-1} \sqrt{\textfrak{p}_\nu}\,\kets{\nu}.
\end{align}
This state has the same $\textswab{p}$, so by Theorem 1 it has the same performance as the original one.
Moreover, because the probabilities $p_n$ that $\ketsi{\psi}$  associates with $n \geq M$ all contribute to photon numbers less than $M$ in $\kets{\psi}$, the total signal energy in \eqref{singlemodepsi}, as given by Eq.~\eqref{signalenergy}, can only be equal to or lower than that of the original probe state.
Thus, \eqref{singlemodepsi} provides the same performance at equal or lesser signal energy than
any state with the same $\textswab{p}$.
\\\\
\indent (\textit{b}) If $N_S\geq (M-1)/2$, consider the probe state
\begin{align}
\label{zeroerrorstate}
  \kets{\psi} = \frac{1}{\sqrt{M}} \sum_{\nu=0}^{M-1} \kets{\nu},
\end{align}
which has energy $(M-1)/2$ and is therefore allowed by the energy constraint. From \eqref{grammatrixelement}, we see that 
$G_{mm'} = \delta_{m,m'}$  so that the output states are mutually orthogonal and the error probability
is zero.
\end{proof}
\end{thm}

The conclusion of Theorem $2(b)$ that discrimination with zero error is possible whenever $N_S \geq (M-1)/2$ is remarkable. In a communications framework, it implies that if a signal energy of at least $(M-1)/2$ is available, we can communicate one of $M$ messages without error using phase modulation of the state \eqref{zeroerrorstate}. This conclusion was noted in \cite{HallFuss91}, although this, or any other, single-mode state cannot achieve error-free communication under the restricted class of measurements allowed in \cite{HallFuss91}. In \cite{Shapiro93}, an alternative scheme using two signal modes suffering conjugate phase shifts of $\theta_m$ and $-\theta_m$ respectively was proposed that achieves the same end. 

Recently, it was shown in \cite{Hirota11} that using the probe state 
\begin{align} 
\label{ECS}
    \ket{\psi}_{ECS} 
  = \frac{1}{\mathcal{N}}\left(\keti{\alpha}\kets{\alpha}-\keti{{-\alpha}}\kets{{-\alpha}} \right), 
\end{align}
for reading a binary phase-encoded memory (so that the output states correspond to Eq.~\eqref{output states} with $M=2$) results in zero-error discrimination. Here, $\kets{\pm \alpha}$ and $\keti{\pm \alpha}$ are coherent states (with $\alpha \neq 0$ but otherwise arbitrary) and the normalization factor $\mathcal{N} = \sqrt{2 (1 - e^{-4|\alpha|^2})}$. The state \eqref{ECS} is an example of an \emph{entangled coherent state} (ECS) \cite{ECS}. 

The existence of several states allowing zero-error discrimination is in itself not surprising if we apply Theorem 1 to the state \eqref{zeroerrorstate}, because according to that theorem, any state with uniform $\textswab{p} = \left(1/M, \ldots, 1/M \right)$ must provide zero-error discrimination and we can clearly write down an infinite number of states with uniform $\textswab{p}$. A more interesting question is whether there exist states with nonuniform $\textswab{p}$, or with signal energy less than $(M-1)/2$, that also allow zero-error discrimination. Theorem $3(b)$ below implies that the answer to both questions is negative, so that a signal energy of at least $(M-1)/2$ is a \emph{necessary condition} for zero-error discrimination. In this connection, it may be verified that the ECS of Eq.~\eqref{ECS} has $\textswab{p} = (1/2,1/2)$ and signal energy $|\alpha|^2/\left(2 \tanh |\alpha|^2\right) > 1/2$ for $|\alpha| > 0$ \cite{footnoteECS}.

The next result gives the form of the optimal probe state for a signal energy constraint
$N_S < (M-1)/2$ and also shows that the only probe states achieving zero-error discrimination have uniform $\textswab{p}$, and therefore have signal energy greater than or equal to $(M-1)/2$.

\begin{thm}
\begin{enumerate}[(a)] \item Among all probe states satisfying 
$\langle\hat{N}_S\rangle\leq N_S < (M-1)/2$, the minimum error probability is achieved by the state 
\begin{align} 
\label{opt}
\ket{\psi}_{\mathrm{opt}} = \sum_{\nu=0}^{M-1} \sqrt{\textfrak{p}_{\nu}}\,\kets{\nu},
\end{align}
with $\textswab{p}$ given by 
\begin{align}
  \textfrak{p}_\nu = \frac{1}{\left(A + \nu B\right)^2}, \;\nu \in \mathbb{Z}_M, 
\end{align}
where $A$, $B$ are positive constants chosen to satisfy the constraints
\begin{equation}
\sum_{\nu=0}^{M-1} \textfrak{p}_{\nu} = 1,
   \quad
\sum_{\nu=0}^{M-1} \nu\, \textfrak{p}_{\nu} = N_S.
\end{equation}
\item Any probe state achieving zero-error discrimination must have $\textswab{p} = (1/M, \ldots, 1/M)$ and signal energy greater than or equal to $(M-1)/2$.
 \end{enumerate}

\end{thm}
\begin{proof}
(\textit{a}) By Theorem $2(a)$, it suffices to consider a single-mode probe state of the form \eqref{singlemodepsi}. As shown in Section II, the output states $\{\ketri{\psi_m}\}_{0}^{M-1}$ form a symmetric set in the sense of having equal a priori probabilities and satisfying Eqs.~(\ref{symmetric set condition.I}) and (\ref{symmetric set condition.II}). It was shown in Refs.~\cite{Belavkin75,Ban97} that the Square-Root Measurement is the minimum error probability measurement for any symmetric pure-state set. 
An explicit formula exists for this minimum error probability \cite{Belavkin75,Ban97,Kato99}. We use the following expression from ref.~\cite{Kato99}:-  
\begin{align}
\label{Peformula}
  \overline{P}_e = 1 - \frac{1}{M^2}\left(\sum_{m=0}^{M-1} \sqrt{\lambda_m}\right)^2,
\end{align}
where $\boldsymbol{\lambda}=(\lambda_0, \ldots, \lambda_{M-1})$ is the vector of eigenvalues
of the Gram matrix $\mathbf{G}$ \cite{footnote3}. The ordered vector of eigenvalues $\boldsymbol{\lambda}$ is specified by the formula (cf. Eq.~(42) of \cite{Kato99})
\begin{align} \label{lambda}
  \lambda_{m'}
   &= \sum_{m=0}^{M-1} \braketir{\psi_0}{\psi_m}\; \cdot e^{-i m'm\theta_M} \\
   &= \sum_{m=0}^{M-1} G_{0m}\, e^{-i m'm\theta_M}
      \label{eigvalue2}.
\end{align}
Note that the eigenvalues only depend on the first row of the Gram matrix -- indeed, the
symmetric set property of the $\{\ketri{\psi_m}\}$ guarantees that the remaining rows are
obtained by cyclically shifting the first. While the ordering of eigenvalues in $\boldsymbol{\lambda}$ enables writing down the compact formula \eqref{lambda} and has a physical interpretation that will appear, note that the error probability itself does not depend on the ordering. 

We may rewrite \eqref{eigvalue2} as
\begin{align} \label{eqn:lambda}
\boldsymbol{\lambda} = \mathcal{F}\left[ \mathbf{G}_0\right],
\end{align} where $\mathcal{F}$ is the Discrete Fourier Transform (DFT) on $\mathbb{Z}_M$ 
and $\mathbf{G}_0 \equiv \{G_{0m}\}$ is  the first row of the Gram matrix. On the other 
hand, \eqref{grammatrixelement} implies that 
\begin{align}
\mathbf{G}_0 = M \cdot\mathcal{F}^{-1} \left[\textswab{p}\right],
\end{align}
where $\mathcal{F}^{-1}$ is the inverse DFT. We therefore have
\begin{align}
  \boldsymbol{\lambda} = M\,\textswab{p},
\end{align}
which gives a physical interpretation for  $\boldsymbol{\lambda}$. 

Minimizing $\overline{P}_e$ for states of the form \eqref{singlemodepsi} with 
signal energy $\langle\hat{N}_S\rangle = N_S$ is then equivalent 
to maximizing the concave function $\sum_{\nu=0}^{M-1} \sqrt{\textfrak{p}_{\nu}}$ over 
the convex set of $\textswab{p}$'s for which the probability normalization constraint
\begin{equation}
\sum_{\nu=0}^{M-1} \textfrak{p}_{\nu} = 1
\end{equation}
and the signal energy constraint
\begin{equation}
\sum_{\nu=0}^{M-1} \nu \, \textfrak{p}_{\nu} =  N_S
\end{equation}
are satisfied. Following the usual Lagrange multiplier method, we define
\begin{align}
  & F(\textfrak{p}_{0},\ldots,\textfrak{p}_{M-1}, A,B) \nonumber \\
  &=  \sum_{\nu=0}^{M-1} \sqrt{\textfrak{p}_{\nu}} - A\left(\sum_{\nu=0}^{M-1} \textfrak{p}_{\nu} - 1\right)
     - B\left(\sum_{\nu=0}^{M-1} \nu\,   \textfrak{p}_{\nu} - N_S\right),
\end{align}
and solve the equation $\nabla F = 0$.   The solution is
\begin{align}
\label{eqn:optimal}
 \textfrak{p}_{\nu} = \frac{1}{(A+\nu B)^2}, \;\;\nu \in \mathbb{Z}_M,
\end{align}
where $A, B$ are chosen such that $\sum_{\nu=0}^{M-1} \textfrak{p}_{\nu} = 1$ and
$\sum_{\nu=0}^{M-1} \nu \, \textfrak{p}_{\nu} = N_S$. The point $\textswab{p}$ defined 
by \eqref{eqn:optimal} is an interior point of the domain of optimization and a local 
maximum by the gradient condition. Since the function being maximized is concave, it is 
also a global maximum on the domain \cite{BoydVandenberghe04}. Thus, the state 
\eqref{eqn:optimal} achieves the minimum error probability among  probe states with energy exactly $N_S$.
Lemma 1, below, establishes that the state \eqref{eqn:optimal} is also optimal 
under the inequality constraint \eqref{signalenergyconstraint}.
\\\\
\indent(\textit{b}) From \eqref{Peformula}, it is evident that $\overline{P}_e=0$ 
if only if $\sum_{\nu=0}^{M-1} \sqrt{\textfrak{p}_{\nu}} = \sqrt{M}$.  It is easy to verify that the maximum value of the quantity $\sum_{\nu=0}^{M-1} \sqrt{\textfrak{p}_{\nu}}$ under just the constraint $\sum_{\nu=0}^{M-1} \textfrak{p}_{\nu} =1$ is $\sqrt{M}$ and is achieved only if $\textfrak{p}_\nu = 1/M, \; 0 \leq \nu \leq M-1$. Thus, the states providing zero-error discrimination are exactly those with uniform $\textswab{p}$. Such a state has signal energy at least $(M-1)/2$. 
\end{proof}
 
The following result may be expected on physical grounds (although it does 
not hold for $N_S > (M - 1)/2$), but we need a proof because this fact is used to complete the proof of Theorem 2(a).
\begin{lem} 
The optimum single-mode probe state of the form \eqref{singlemodepsi} under the inequality constraint 
$\langle \hat{N}_S \rangle \leq N_S < (M - 1)/2$ is the same as the 
optimum state under the equality constraint 
$\langle \hat{N}_S \rangle = N_S$.
\begin{proof} 
Consider the maximization of  $\sum_{\nu=0}^{M-1} \sqrt{\textfrak{p}_{\nu}}$ under the probability constraint  $\sum_{\nu=0}^{M-1} {\textfrak{p}_{\nu}}  = 1$. As shown in the proof of Theorem $3 (b)$, the maximum is achieved at $ \textswab{p}^* = (1/M, \ldots, 1/M)$. Denote the $\textswab{p}$ of an optimal state under the inequality constraint by $\textswab{p}_*$. Denote its signal energy by $N_s \leq N_S$. Since the set of all $\textswab{p}$ is convex, the line segment $L$ joining $\textswab{p}_*$ 
to $\textswab{p}^*$  in that set consists of allowed $\textswab{p}$'s. Since the signal energy is a linear function of $\textswab{p}$, $L$ contains states of signal energy ranging from $N_s$ to $(M-1)/2$. Further, since the function $\sum_{\nu=0}^{M-1} \sqrt{\textfrak{p}_{\nu}}$ is a concave function of $\textswab{p}$ whose maximum is attained at $\textswab{p}^*$, the function must be nondecreasing as we move along $L$ from $\textswab{p}_*$ to $\textswab{p}^*$ \cite{BoydVandenberghe04}. In particular, we can find a state on $L$ with signal energy $N_S$ and equal or better performance than that obtainable from $\textswab{p}_*$. Consequently, there is an optimal state under the inequality constraint with signal energy exactly $N_S$. One such state must be that given by \eqref{eqn:optimal}.
\end{proof}
\end{lem}

The optimum probe state for binary phase discrimination is particularly straightforward.
\begin{cor}[Binary case]
For $M=2$, the optimum probe state for $N_S < 1/2$ is
\begin{align} \label{binoptstate}
\ket{\psi} = \sqrt{1-N_S}\,\kets{0} + \sqrt{N_S}\, \kets{1}.
\end{align}
\begin{proof}
This follows immediately by solving for $A$ and $B$ in Theorem 3.
\end{proof}
\end{cor}
However, note that our proposed implementation of the binary case for
achieving the minimum error probability
\begin{align} \label{binerrorprob}
   \overline{P}_e = 1/2 - \sqrt{N_S(1 - N_S)}
\end{align}
in Section IV uses a different probe state (with the same signal energy and performance) for practical reasons.
 
In the general ($M > 2$) case, it appears that closed-form solutions for $A$ and $B$ appearing in the expression \eqref{eqn:optimal} for the
optimal state cannot be obtained, so recourse to numerical evaluation becomes necessary. Simulations of the resulting performance are presented in Section III.E.

\subsection{Optimal measurement}

Let us briefly discuss the quantum measurement, as determined by the corresponding POVM, that optimally distinguishes the output states $\{\ketri{\psi_m}\}_{m=0}^{M-1}$. Following refs.~\cite{Kato99,Ban97,Belavkin75}, the optimum POVM, the SRM, is a rank-one measurement with elements $\hat{\Pi}_{m} = \ket{\chi_m}\brari{\chi_m}, \, m \in \mathbb{Z}_M$ ($\ketri{\chi_m}$ may have norm less than one) with
\begin{align} \label{optimumPOVM}
\ketri{\chi_m} = \left(\sum_{n=0}^{M-1} \ket{\psi_n} \brari{\psi_n}\right)^{-1/2} \ketri{\psi_m},
\end{align}
where the operator in parentheses (and its inverse) is defined on just the span of $\{\ketri{\psi_m}\}_{m=0}^{M-1}$. For a single-mode probe of the form $\kets{\psi}= \sum_{n=0}^{M-1} \sqrt{\textfrak{p}_n} \, \kets{n}$ of Theorem 2, we have
\begin{align}
\ketr{\chi_m} = \frac{1}{\sqrt{M}}\sum_{n \,:\, \textfrak{p}_n \neq 0} e^{imn\theta_M}\, \ketr{n}.
\end{align}
Note that the optimum measurement is the same for any two single-mode 
probe states $|\psi\rangle_S,|\psi'\rangle_S$ with $\textswab{p}, \textswab{p}'$ having the same support.
The optimum state of \eqref{opt} has no zero coefficients so that the optimum measurement elements are
\begin{align} 
\label{optmeas}
  \ketr{\chi_m} 
  = \frac{1}{\sqrt{M}}\sum_{n =0}^{M-1} e^{imn\theta_M}\, \ketr{n}.
\end{align}
Since these vectors form an orthonormal set, the measurement is a projective (von Neumann) measurement. Indeed, these measurement vectors coincide with the eigenstates of the unitary Pegg-Barnett phase operator \cite{PeggBarnett} on the truncated Hilbert space 
$\cl{H}_{M-1} = \mathrm{span}\,\{\ketr{0}, \ldots, \ketr{M-1}\}$.

\subsection{Combined energy constraint}
We have so far assumed that a constraint is placed on the energy of just the signal modes. In some situations, it may make sense to constrain -- as a measure of all the resources involved in state preparation --  the average total energy in the signal and idler modes \emph{combined} without restricting either individually. In other words, we impose the constraint

\begin{align}
  \left\langle\hat{N}\right\rangle &:= \left\langle\hat{N}_S + \hat{N}_I\right\rangle \nonumber\\
  &:= \left\langle{\sum_{j=1}^J \hat{N}_S^{(j)}}\right\rangle + \left\langle{\sum_{j'=1}^{J'} \hat{N}_I^{(j')}}\right\rangle
  \leq N,
\end{align}
for a given number $N$,  where $\{\hat{N}_S^{(j)}\}$ and $\{\hat{N}_I^{(j')}\}$ are the modal signal and idler photon number operators respectively, and ask for a state satisfying this constraint that minimizes the error probability.  This new problem has the same solution that we found before. 
\begin{thm}
Among pure-state probes, the minimum error probability achievable under a combined energy constraint $N$ is identical to the minimum error probability achievable under a signal energy constraint $N$. 
\begin{proof}
The combined energy constraint of $N$ is clearly more restrictive than a 
signal energy constraint of $N$.  We showed in Section III.A that the optimal
pure-state probe for a signal energy constraint of $N$ is a signal-only
state of energy $N$.  Since this state also has combined energy $N$, 
it remains the optimal state under the combined energy constraint.
\end{proof}
\end{thm}

\subsection{Mixed probe states}

In our work so far, we have assumed that the probe state $\ketsi{\psi}$ was pure. Of course, we may also use a mixed probe state $\hat{\rho}_{IS}$ resulting in the mixed output states
\begin{align} \label{mixedoutputstate}
\hat{\rho}_m = \left(\hat{I}_I \otimes \hat{U}_m\right) \hat{\rho}_{IS} \left(\hat{I}_I \otimes \hat{U}^{\dag}_m \right)
\end{align}
for $\hat{U}_m$ given by Eq.~\eqref{unitary.m}. We now show that allowing for mixed probes does not lead to improved performance. We actually prove a stronger result.

\begin{thm}
Let $\hat{\rho}_{IS}$ be a mixed state with ensemble decomposition 
$\hat{\rho}_{IS} = \sum_j \pi_j |\psi_j\rangle_{IS}\langle\psi_j|$ and 
with signal energy $\mathrm{tr}(\hat{\rho}_{IS}\hat{N}_S)\leq N$ or 
with combined energy $\mathrm{tr}(\hat{\rho}_{IS}\hat{N})\leq N$.
A transmitter preparing the ensemble $\{|\psi_j\rangle_{IS}\}$ with 
probabilities $\{\pi_j\}$ and a receiver making optimal measurements
conditioned on knowledge of $j$ cannot beat the performance of the 
optimal pure-state probe under either energy constraint.
\begin{proof} 
Let $\overline{P}_e[\cdot]$ denote the minimum error probability
attainable on using the argument as probe state.  
We have the chain 
of inequalities
\begin{align}
  \overline{P}_e[\hat{\rho}_{IS}]
   &\geq \sum_j \pi_j \overline{P}_e[|\psi_j\rangle_{IS}]  
        \label{eqn:ln1} \\
   &= \sum_j \pi_j \overline{P}_e[|\psi_j^*\rangle_S] 
        \label{eqn:ln2} \\
   &\geq \overline{P}_e[|\overline{\psi}\rangle_S] 
        \label{eqn:ln3} \\
   &\geq \overline{P}_e[|\psi^{\mathrm{opt}}\rangle_S].
\end{align}
In \eqref{eqn:ln1}, the right-hand side represents the optimum performance
under the conditions of the theorem statement and the inequality holds
because the performance given by the left-hand side is obtained when knowledge of 
$j$ is ignored by the receiver \cite{notemixed}.  
$|\psi_j^*\rangle_S$ is the state of the form \eqref{singlemodepsi} with the same \textswab{p} as $|\psi_j\rangle_{IS}$, denoted $\textswab{p}_j$,
and \eqref{eqn:ln2} holds by Theorem 1.  In \eqref{eqn:ln3}, 
$|\overline{\psi}\rangle_S$ is the state of the form \eqref{singlemodepsi}
with $\textswab{p} = \sum_j \pi_j \textswab{p}_j$, and the inequality
is true because $\overline{P}_e$ of \eqref{Peformula} is a sum of 
convex functions of $\textswab{p}$ \cite{notemixed2}, and hence convex itself.  None of the state transformations above has increased the signal
(or combined) energy from that of $\hat{\rho}_{IS}$, so that 
$|\overline{\psi}\rangle_S$ is a pure state with energy bounded by $N$.
It cannot beat the optimum pure state
$|\psi^{\mathrm{opt}}\rangle_S$ with energy $N$.
\end{proof}
\end{thm}

With Theorem 5 in hand, we may conclude that the state of Theorem 3 is in fact the pure or mixed state of energy at most $N_S$ that has the lowest error probability. This result has the following implication. Recall that a \emph{classical state} of the signal and idler modes is a density operator $\rho_{IS}$ expressible as a mixture of coherent states in the form
\begin{align} \label{classical}
\rho_{IS} =  \int   P(\boldsymbol{\alpha},\boldsymbol{\beta}) \keti{\boldsymbol{\alpha}} \kets{\boldsymbol{\beta}}\brai{\boldsymbol{\alpha}} \bras{\boldsymbol{\beta}} \,d^2\boldsymbol{\alpha} \, d^2\boldsymbol{\beta}
\end{align}
where $\keti{\boldsymbol{\alpha}}$ ($\kets{\boldsymbol{\beta}}$) is a  multimode idler (signal) coherent state and $P(\boldsymbol{\alpha},\boldsymbol{\beta}) \geq 0$ is a probability density \cite{GerryKnight05}. Such states are readily prepared from laser outputs using beam splitters and classical random numbers, and standard measurements made on them are quantitatively describable using semiclassical photodetection theory \cite{Shapiro09}. 
The optimal state \eqref{opt} is a nonclassical state (coherent states are the only pure classical states), and being optimal, it performs better than the coherent state of energy $N_S$ (see Fig.~ 3). However, from Theorem $3$ alone, we cannot conclude that it performs better than an arbitrary classical state of the form \eqref{classical}. With  the addition of Theorem $4$, we can draw the conclusion that the nonclassical state 
\eqref{opt} outperforms \emph{all} classical states. Note that our argument does not imply that the coherent state of energy $N_S$ is the optimal classical state with energy $N_S$. 

\subsection{Numerical results}

In this section, we numerically compare the performance of the optimal state of Theorem $3$ with that of some standard states in quantum optics. We first consider some signal-only states and then a couple of two-mode entangled states. With the exception of the optimum state \eqref{opt}, the performance curves are obtained by computing the first row $\mathbf{G}_0$ of the Gram matrix, for which analytical formulas are given below for each family of states. We then compute the eigenvalue vector $\boldsymbol{\lambda}$ (or equivalently $\textswab{p}$) using an FFT routine (cf. Eq.~\eqref{eqn:lambda}), and finally the error probability $\overline{P}_e$ via Eq.~\eqref{Peformula}. 

\subsubsection{Optimal state}

For each value of the probe signal energy $N_S$, the $\textswab{p}$ of the optimal state is first obtained by numerically solving for $A$ and $B$ appearing in \eqref{eqn:optimal}. We then compute the error probability via \eqref{Peformula} and plot it against $N_S$ in Fig.~3 for a few values of $M$. 
We find that as $N_S \to (M-1)/2$, $A \to \sqrt{M}$ and $ B \to 0$, so that the numerically computed optimal state of \eqref{eqn:optimal} approaches the uniform superposition state of Theorem $2 (b)$. The approach to zero error in the same limit is clearly visible in Fig.~3. For comparison, the coherent state performance (See section III.E.2) is also plotted alongside. 

\begin{figure}
\centering
\includegraphics[trim= 30mm 80mm 30mm 84mm, clip=true, scale=0.55]{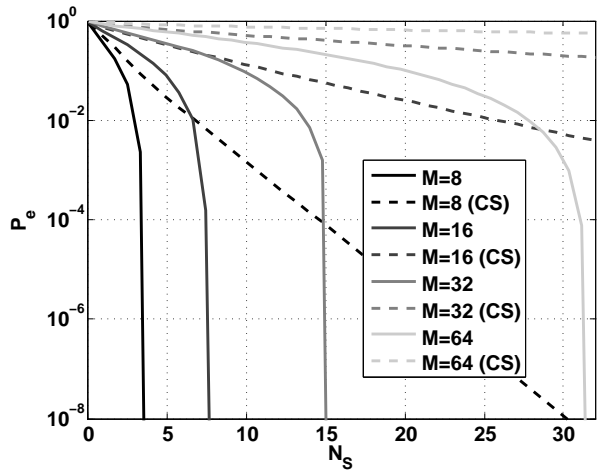}
\caption{\small The error probability $\overline{P}_e$ of the optimum probe given by Eq.~\eqref{opt} (solid) and the coherent-state probe (dashed) as a function of $N_S$ for  $M=8, 16, 32,$ and 64. Curves for larger $M$ are shown lighter.}
\end{figure}

\subsubsection{Squeezed state and coherent state}

The single-mode squeezed states \cite{Yuen76} are a well-known class of states. 
Let $\hat{a}_S = \mu^*\hat{b} - \nu\hat{b}^\dagger$, where $|\mu|^2 - |\nu|^2 = 1$
and $\hat{b}$ is in a 
coherent state $|\alpha\rangle$ satisfying $\hat{b}|\alpha\rangle = \alpha|\alpha\rangle$, 
$\alpha > 0$.
Then, the signal mode $\hat{a}_S$ is in the squeezed state (called two-photon coherent state (TCS) in \cite{Yuen76}) $|\alpha;\mu,\nu\rangle$.
In the following, we assume $\mu,\nu$ are real and $\mu > 0$.
If $\nu<0$,  this transformation corresponds to squeezing in the imaginary (phase) quadrature while stretching the real quadrature.  The average energy of this state may be calculated to be
\begin{align}
    \left\langle \hat{a}_S^\dag \hat{a}_S \right\rangle 
  = \left(\mu - \nu\right)^2 \alpha^2 + \nu^2.
\end{align}
The action of the $m$-th phase $\theta_m$ takes $\hat{a}_S$ to 
$\hat{a}_R^{(m)} = \hat{a}_S e^{i\theta_m}$ in the Heisenberg picture so that
\begin{align}
     \hat{a}_R^{(m)} 
   = \mu e^{i\theta_m} \hat{b} - \nu e^{i\theta_m} \hat{b}^\dag.
\end{align}
Thus, the Schr\"{o}dinger-picture state at the end of these transformations is 
$\ketr{\alpha; \mu e^{-i\theta_m}, \nu e^{i\theta_m}}$.

We may now use Eq.~(3.25) of ref.~\cite{Yuen76} to write down the Gram matrix elements $G_{0m}$:-
\begin{align} \label{SSG0}
  G_{0m} &=
    (\mu^2 - \nu^2 e^{2i\theta_m})^{-1/2} \notag\\
    &\quad
     \times \exp\left[
         \alpha^2
         \Bigl(
            \frac{e^{i\theta_m} - \mu\nu(e^{2i\theta_m} - 1)}{\mu^2 - \nu^2 e^{2i\theta_m}}  - 1
         \Bigr)
          \right].
\end{align}

The \emph{coherent state} $\kets{\alpha}$ is identical to the TCS $\ketr{\alpha; 1, 0}$ and its Gram matrix elements can be obtained from eq.~\eqref{SSG0}.

\begin{figure}
\includegraphics[scale=0.43]{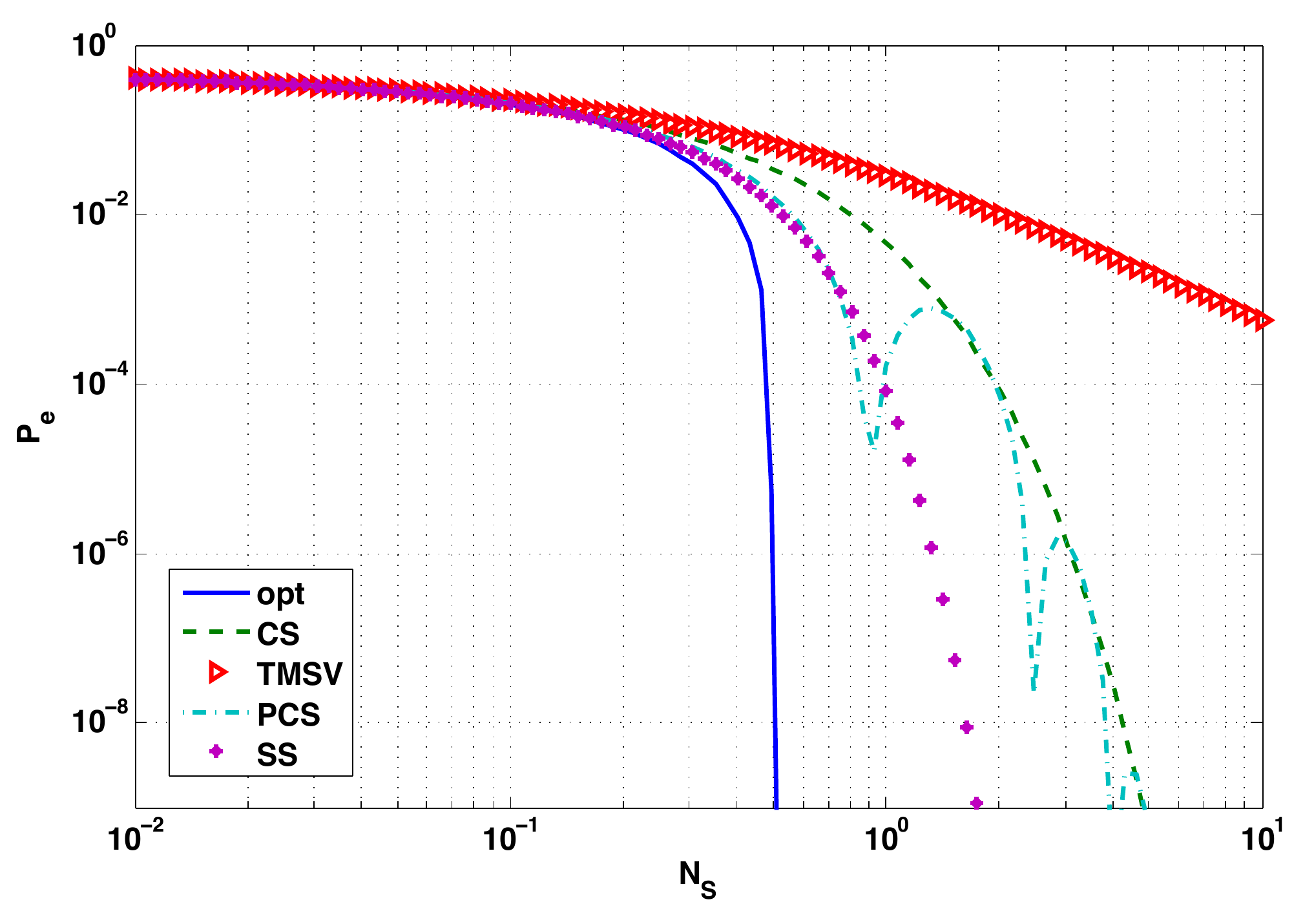}
\caption{(Color online) Minimum error probability as a function of signal energy  $N_S$ for the states of Section III.E for 
the binary case $M = 2$. The squeezed state has squeeze
parameters $\mu,\nu$ optimized for each $N_S$, and the two-mode squeezed vacuum has $J = 1$.
}
\end{figure}
\subsubsection{Two-mode squeezed vacuum state}

The two-mode squeezed vacuum (TMSV) state with signal energy $N_S$ is the two-mode state \cite{GerryKnight05}:-
\begin{align} \label{TMSV}
\ketsi{\psi^{\mathrm{TMSV}}(N_S)} = \sqrt{\frac{1}{N_S+1}}\;\sum_{n=0}^{\infty} \sqrt{\frac {N_S^n} {(N_S+1)^n}} \keti{n}\kets{n}.
\end{align}
We may also consider using $J$ copies of a TMSV state as probe. In order to keep the same signal energy $N_S$, we must use $J$ copies of $\ketsi{\psi^{\mathrm{TMSV}}(N_S/J)}$. We may directly compute the first row of the Gram matrix of the output states as
\begin{align} 
\label{TMSVG0}
  G_{0m} =
    \left[ 
       \frac{1}{1 + \frac{N_S}{J}(1 - e^{i m \theta_M})}
    \right]^J, \; \;\; m \in \mathbb{Z}_M.
\end{align}

\subsubsection{Pair-coherent state}

The pair-coherent states (PCS) are a family of two-mode states parametrized by $\zeta \in \mathbb{C}$ and a non-negative integer $q$  \cite{PCSrefs}. We will consider the case $q=0$ corresponding to equal energy in the signal and idler modes.  Such a PCS has the form
\begin{align} \label{PCS}
\ketsi{\psi^{\mathrm{PCS}}(\zeta)} = \frac{1}{\sqrt{I_0(2|\zeta|)}}\;\sum_{n=0}^{\infty} \frac {\zeta^n} {n!} \keti{n}\kets{n},
\end{align}
where $I_0(\cdot)$ is the modified Bessel function of first kind and order zero. The signal energy $N_S$ is related to $\zeta$ via 
\begin{align}
N_S = \frac{|\zeta| I_1(2|\zeta|)}{I_0(2|\zeta|)},
\end{align}
where $I_1(\cdot)$ is the modified Bessel function of first kind and order one. Since the phase of $\zeta$ does not affect the performance, we will assume $\zeta$ to be real and positive. The Gram matrix elements may be computed to be
\begin{align} \label{PCSG0}
G_{0m} =\frac{I_0(2\zeta e^{im\theta_M/2})}{I_0(2\zeta)}, \;\;\; m \in \mathbb{Z}_M.
\end{align}


\subsubsection{Performance curves}
Numerical results for $M = 2$ and $M = 8$ are shown in Figs. 4 and 5 
respectively.  For the special case of binary discrimination, the minimum error 
probability is
\begin{equation}
  P_e = \frac{1}{2}\left(1 - \sqrt{1 - |\sigma|^2}\right),
\end{equation}
where $\sigma = {}_{IR}\langle\psi_0|\psi_1\rangle_{IR}$.
For each of the probe states considered above, we have
\begin{align}
  \sigma_{\mathrm{CS}} &= e^{-2N_S} \\
  \sigma_{\mathrm{SS}} &= e^{-2N_S(N_S + 1)} \\
  \sigma_{\mathrm{TMSV}} &= \frac{1}{(1 + 2N_S/J)^J} \\
  \sigma_{\mathrm{PCS}} &= \frac{J_0(2\zeta)}{I_0(2\zeta)}.
\end{align}
For the squeezed-state (SS) probe, an optimal squeezing has been 
assumed (see below).  For
multiple copies of two-mode squeezed vacuum, note that $\sigma_{\mathrm{TMSV}}$ is 
decreasing in $J$ and $\sigma_{\mathrm{TMSV}}\to \sigma_{\mathrm{CS}}$, as $J\to\infty$.  For certain values of $N_S$, the PCS performs best among
all but the optimal state.

\begin{figure}
\includegraphics[scale=0.4]{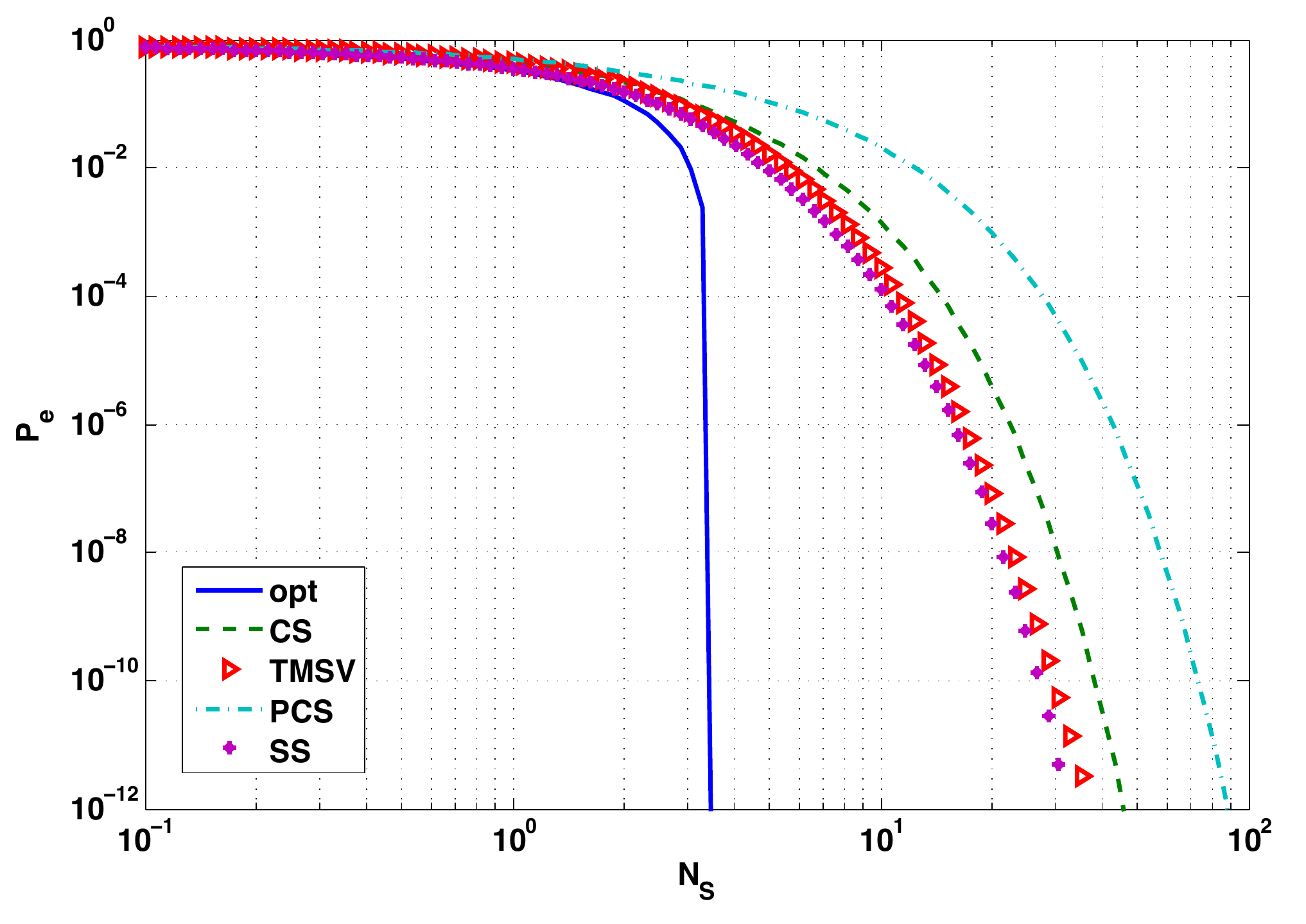}
\caption{(Color online) Minimum error probability as a function of signal energy  $N_S$ for the states of Section III.E for 
$M = 8$.  The squeezed state has $f = \nu^2/N_S$ optimized for each $N_S$ and the two-mode squeezed vacuum uses
 $J = 15$ signal modes.}
\end{figure}

For $M=8$, on the other hand, the PCS performs the worst and is consistently beaten by the coherent state. The squeezed state with an optimized amount of squeezing performs the best among the suboptimal states and is closely matched by the two-mode squeezed vacuum, both of which consistently beat the coherent state.

For the squeezed-state probe, we can further discuss the optimal squeeze parameters $\mu,\nu$ for a given energy $N_S$.
For symmetric phase discrimination, we might assume that phase-squeezing,
$\mu > 0, \nu < 0$,
is optimal.  The following is what we find.
\begin{itemize}
\item $M = 2$.  The optimal $\nu = N_S/\sqrt{1 + 2N_S}$ is positive.  The phase space representation is shown in Fig.~$6(a)$ for this case.
Note that for a given $N_S$, the mean $\langle\hat{a}_R\rangle$ is independent of the sign of $\nu$.  

\item $M = 3$.  Numerically, we find the optimal $\nu$ is positive. 
\item $M = 4$.  Numerically, no squeezing $\nu = 0$ (coherent state) appears to be optimal. 

\item $M \geq 5$.  Numerically, we find the optimal $\nu$ is negative. 
For large $M$, phase-squeezing can be seen to reduce the overlap between
neighboring states (Fig.~$9 (b)$). 
\end{itemize}

\begin{figure}
\includegraphics[scale=0.65]{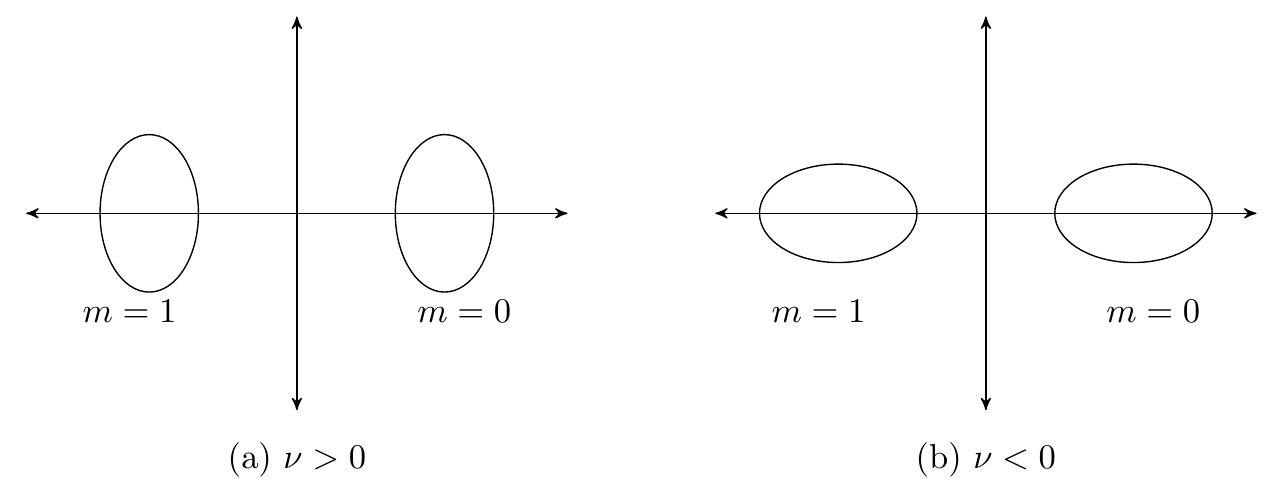}
\caption{\small Phase space representation of
$M = 2$ squeezed-state probe: (a) Amplitude-squeezed.  (b) Phase-squeezed.}
\end{figure}

\begin{figure}
\includegraphics[scale=0.65]{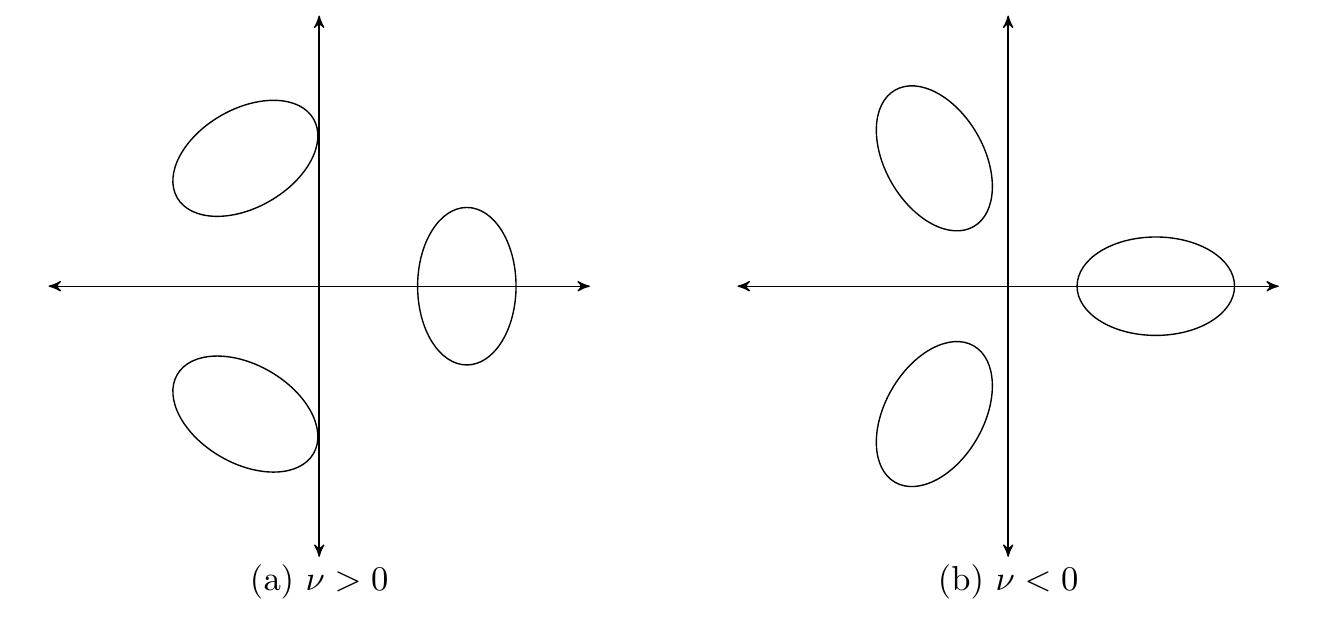}
\caption{\small Phase space representation of
$M = 3$ squeezed-state probe: 
(a) Amplitude-squeezed.  (b) Phase-squeezed.}
\end{figure}

\begin{figure}
\includegraphics[scale=0.65]{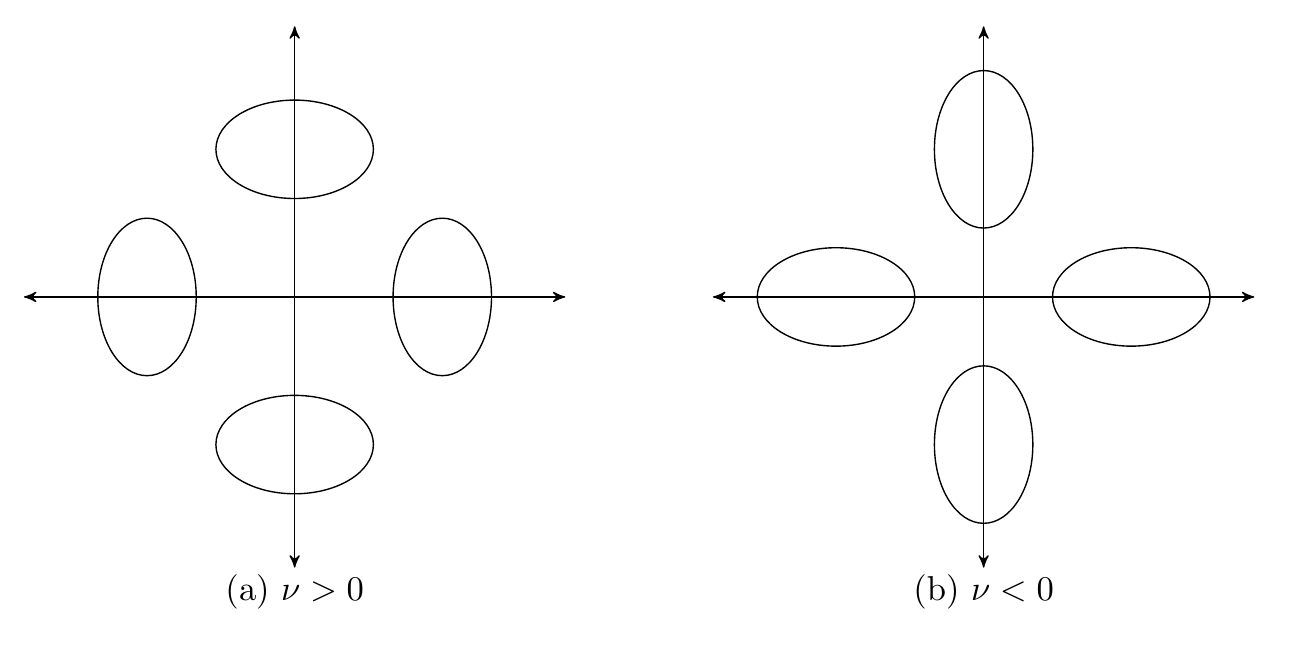}
\caption{\small Phase space representation of
$M = 4$ squeezed-state probe: 
(a) Amplitude-squeezed.  (b) Phase-squeezed.}
\end{figure}

\begin{figure}
\includegraphics[scale=0.65]{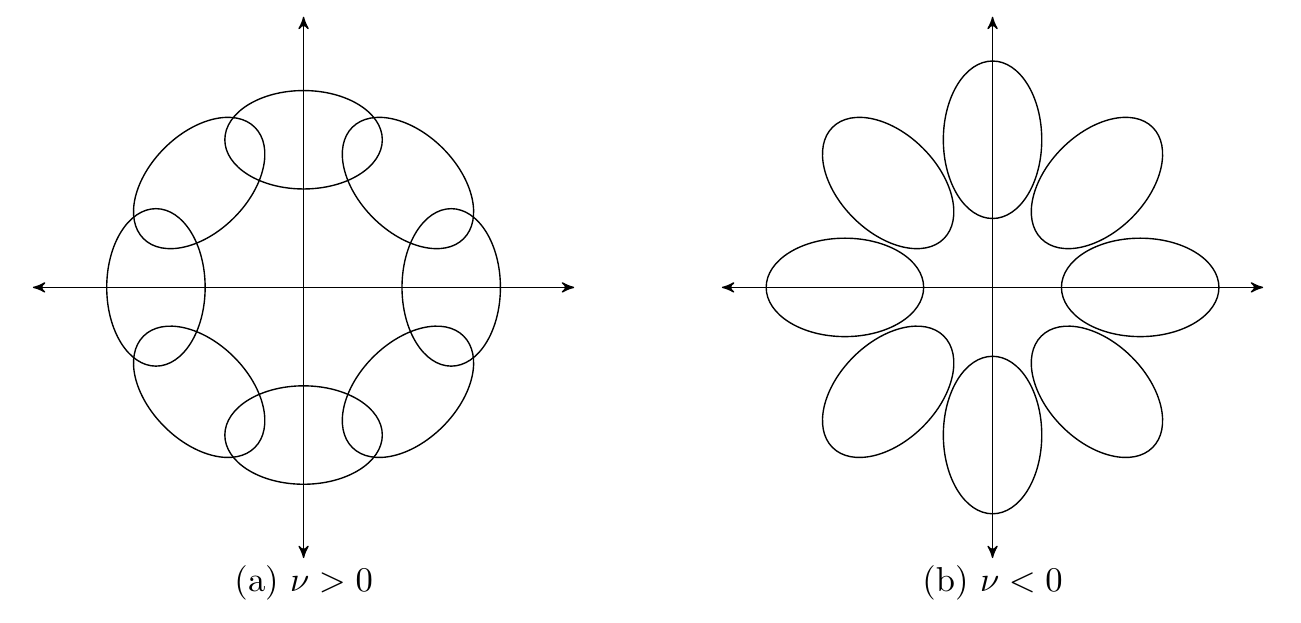}
\caption{\small Phase space representation of
$M = 8$ squeezed-state probe: 
(a) Amplitude-squeezed.  (b) Phase-squeezed.}
\end{figure}

\section{Implementation of Conditionally Optimal Binary Phase Shift Keying}

In this section, we show that the performance of the optimal state for the $M=2$ case, the state of Corollary 1, can be readily obtained in the laboratory with current technology. Furthermore, the inclusion of transmission losses and sub-unity detection efficiencies leads to occasional inconclusive outcomes (or \emph{erasures}) but leaves unchanged the error performance conditioned on no erasure. These results are applicable both to long-distance communication based on binary phase shift keying (BPSK) and to (short-distance) phase sensing or reading of a phase-encoded memory of the type described in ref.~\cite{Hirota11}.

\begin{figure*}
\includegraphics[scale=0.6]{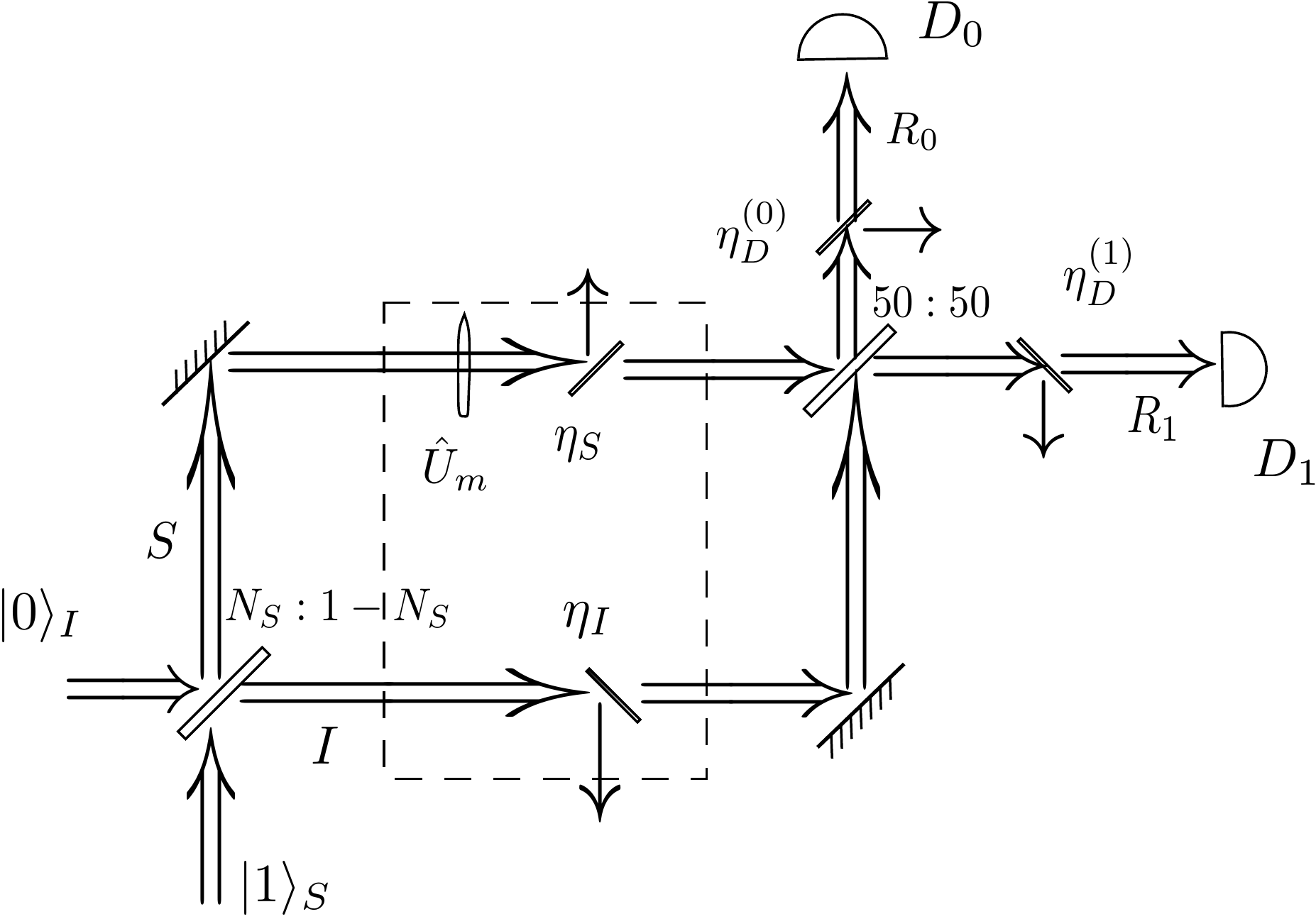}
\caption{\small Setup for realizing the optimal BPSK performance with signal energy $N_S$. The signal ($S$) and idler ($I$) modes are taken to be spatially separated modes with the same polarization and transverse spatial pattern. The $N_S:1-N_S$ beam splitter synthesizes the optimal probe state from a single-photon state. The shaded box containing the phase unitary transformation $\hat{U}_m$ on the signal mode represents the sender's modulator and the transmission medium in a communication scenario or the phase-encoded memory element in a quantum reading scenario. The 50:50 beam splitter performs a rotation of the ideal measurement basis into the single-photon states of the output modes $R_0$ and $R_1$. The small beam splitters with the indicated transmittances model the transmission losses suffered by the signal mode ($\eta_S$), the idler mode ($\eta_I$), and the quantum efficiencies of the detectors $D_0$ and $D_1$. Hypothesis $m$ is declared when detector $D_m$ clicks. If neither detector clicks, an erasure has occurred.}
\end{figure*}

The optimal probe state of Eq.~\eqref{eqn:optimal}, and thus also the $M=2$ optimal state of Eq.~\eqref{binoptstate} (we are assuming $N_S \leq 1/2$ for the latter state), is hard to prepare in a deterministic fashion. However, several techniques exist \cite{QSD,Fiurasek05,Bimbard10} to prepare such Fock-state superpositions in a conditional manner, of which some have been demonstrated experimentally \cite{Bimbard10}. The state \eqref{binoptstate} can be thought of as a qubit state in a ``single-rail'' encoding \cite{Koketal07}. The optimal measurement on the output states following the phase shift is, via Eq.~\eqref{optmeas}, a projective measurement onto the basis $\{\left(\kets{0} \pm \kets{1}\right)/\sqrt{2}\}$. Such a measurement would require implementing a unitary taking the $\{\left(\kets{0} \pm \kets{1}\right)/\sqrt{2}\}$ basis to the $\{\kets{0},\kets{1}\}$ basis, followed by photodetection. However, it is well-known (see ref.~\cite{Koketal07} and references therein) that it is impossible to effect the required unitary transformation deterministically with linear optics, leading to further inefficiencies in an implementation involving the state \eqref{binoptstate}.

By an appeal to Theorem 1, both the state preparation and measurement issues can be circumvented. To do so, we use instead of \eqref{binoptstate} the probe state
\begin{align} \label{altoptstate}
\ketsi{\psi} = \sqrt{1-N_S}\,\ketsi{10} + \sqrt{N_S}\, \ketsi{01},
\end{align} 
where the $S$ and $I$ modes may be spatially separated (as in Fig.~10) or may correspond to orthogonal polarization degrees of freedom of the same spatiotemporal mode. In effect, we have switched to a dual-rail qubit encoding \cite{Koketal07} in going from \eqref{binoptstate} to \eqref{altoptstate}. Note that the states \eqref{binoptstate} and \eqref{altoptstate} have the same $\mathbf{p}$ and $\textswab{p}$ and so have the same signal energy and error performance.

As shown in Fig.~4, the state \eqref{altoptstate} can be prepared by directing a single photon to a $N_S: 1 - N_S$ beam splitter. Further, Eq.~\eqref{optimumPOVM} dictates that the optimum POVM measures the basis $\{\left(\ketsi{01} \pm \ketsi{10}\right)/\sqrt{2}\}$. This may be accomplished by a 50:50 beam splitter followed by single-photon detection using two detectors at the output ports of the beam splitter. Hypothesis  $m \in \{0,1\}$ is declared if detector $D_m$ clicks \cite{polnote}.

 Let us first consider ideal operation and ignore the small beam splitters in Fig.~4 representing transmission and detection losses. It can be verified that the two possible states in the output modes $R_0$ and $R_1$ of the second beam splitter just before the detection are given by
\begin{align} \label{detstates}
\ket{\psi_{0 (1)}}_{R_0R_1} = \lambda_{+(-)} \ket{10}_{R_0R_1} + \lambda_{-(+)} \ket{01}_{R_0R_1}, 
\end{align} 
where
\begin{align}
\lambda_{+ (-)} = \frac{\sqrt{1-N_S} \pm\sqrt{N_S}}{\sqrt{2}}.
\end{align}
One of the detectors always clicks and the error probability evaluates to
\begin{align} \label{binerrorprob2}
\overline{P}_e = \lambda_-^2 = \frac{1}{2} - \sqrt{N_S(1-N_S)},
\end{align}
which agrees with Eq.~\eqref{binerrorprob}.

Let us now consider the performance of the setup of Fig.~4 in a realistic setting that includes transmission loss and non-unity quantum efficiency of the detectors as shown. We assume that the detectors have negligible dark count rates. Transmission loss may be the dominant factor in a long-distance communication system but may be negligible in the reading of a memory. For simplicity, we assume $\eta_S = \eta_I = \eta_T < 1$, which is a realistic assumption in the communication context, although it can be relaxed without altering our conclusions much. We also assume $\eta_D^{(0)} = \eta_D^{(1)} = \eta_D$ and let $\eta = \eta_T \eta_D$. 

By following the evolution of the probe state through the system while preserving unitarity by adding vacuum-state input modes at each of the small beam splitters, we can write the state of the entire system just before the detectors as 
\begin{align}
\ket{\psi_{0(1)}} = &  \sqrt{\eta}\,\left( \lambda_{+(-)}|10\rangle_{R_0R_1} + \lambda_{-(+)} \ket{01}_{R_0R_1}\right) \\ &\times\ket{0000} _{I'S'D_0'D_1'} + \ket{0}_{R_0}\ket{0}_{R_1} \ket{\phi_{\tsf{loss}}}_{{I'S'D_0'D_1'}},
\end{align}
where $S$, $I$, $D_0'$ and $D_1'$ denote the leakage modes and $\ket{\phi_{\tsf{loss}}}_{{I'S'D_0'D_1'}}$ denotes an (un-normalized) state of those modes whose squared norm is $(1-\eta)$ and is an eigenstate of the total photon number in the leakage modes with eigenvalue one. This second term corresponds to the case when none of the detectors click, i.e., an erasure occurs. The probability of erasure $P_\epsilon$ is then independent of $m$ and equals
\begin{align} \label{erasureprob}
P_\epsilon = 1- \eta.
\end{align}

On the other hand, when one of the detectors does click (which happens with probability $\eta$), the error probability of the ensuing decision is exactly the same as before, that given by \eqref{binerrorprob}. Note the erasure probability is independent of $N_S$ and the error probability conditioned on no erasure is independent of $\eta$. This latter probability is identical to the error probability obtainable from the optimum state of Corollary 1. In particular, if $N_S=1/2$, there is no error whenever there is no erasure. \newline

\section{Concluding Remarks}

In this paper, we have set up an $M$-ary phase discrimination problem that naturally models phase-based communication, phase sensing, and quantum reading of a phase-based digital memory.  Allowing for a general entanglement-assisted probing strategy, we have characterized the equivalence classes of probe states with the same performance. We have found the exact form of the optimizing probe state as a function of the energy $N_S$ and characterized the probes that allow zero-error discrimination. From a theoretical point of view, we have thus completely solved a constrained bosonic channel discrimination problem, a class of problems for which exact solutions are rare \cite{Weedbrooketal11}. We have studied the error performance of some standard states in quantum optics. For the $M=2$ case important to reading of a memory, we have shown that the optimal performance can be readily obtained with current technology conditioned on no erasure due to system losses.

From a more practical point of view, the analysis here is limited by not having included the effect of system losses in general. We note that including loss in our problem model brings it into the general lossy image sensing framework considered in ref.~\cite{NairYen11} so that the result of that paper on the form of the optimal probe can be used as a starting point for analysis. Nevertheless, the work of this paper remains essential to the subsequent analysis of the lossy system performance. 

Another serious practical problem is that of synthesizing the optimal probe states in either the lossy or lossless cases as well as realizing the optimal POVMs on them. It may be hoped that the flexibility in state preparation afforded by Theorem 1  can partially alleviate these problems, though it remains to be seen if optimal or near-optimal performance can be achieved in practice for some instances of the problems considered here.

\section{Acknowledgements}

R. N. and B. J. Y. are supported by the Singapore National Research Foundation under NRF Grant No. NRF-NRFF2011-07. Part of this work was accomplished while they were at the Massachusetts Institute of Technology and supported by DARPA's Quantum Sensors Program under AFRL Contract No. FA8750-09-C-0194. The work of S. G. is supported by the DARPA Information in a Photon program under contract No. HR0011-10-C-0162, that of J. H. S. was supported by the DARPA Quantum Sensors Program and the ONR Basic Research Challenge Program, and that of S. P. by EPSRC (EP/J00796X/1).



\begin{thebibliography}{10}
\bibitem{footnote1} A symmetric state set is called a homogeneous state set in \cite{Belavkin75}.
\bibitem{footnote2} For the present problem of discriminating equiprobable symmetric phase shifts, this may be seen from the formulas (\ref{Peformula})-(28) for the error probability. For an arbitrary pure-state ensemble with arbitrary prior probabilities, an argument was sketched in Ref.~[15] of \cite{NairYen11}.
\bibitem{footnoteECS} Thus, interestingly, the signal energy of $\ket{\psi}_{ECS}$ is discontinuous in $|\alpha|$ at $|\alpha|=0$.
\bibitem{footnote3} While the formula \eqref{Peformula} was applied in Ref.~\cite{Kato99} to a case where the output states $\{\ket{\psi_m}\}$ are linearly independent, it can be verified that it is valid in general. 
\bibitem{notemixed} This is a 'perfect-measurement' lower bound
argument such as that used in composite hypothesis testing \cite{VT}.
\bibitem{notemixed2} The geometric mean $\sqrt{x_1x_2}$ is concave
on $\mathbb{R}^+\times\mathbb{R}^+$ \cite{BoydVandenberghe04}.
\bibitem{polnote} The scheme of Fig.~10 can be easily modified in the case where $S$ and $I$ refer to orthogonal polarizations of a single photon in a given spatiotemporal mode, although some care is needed in the design to ensure that the $I$-polarization does not undergo the $0 \,(\pi)$ phase shift that the $S$-polarization does.
\end{thebibliography}
\end{document}